\documentclass[12pt,a4paper]{article}
\usepackage{amsmath}
\usepackage{amsfonts}
\usepackage{amssymb}
\usepackage{amsthm}
\usepackage[pagewise]{lineno}
\usepackage[affil-it]{authblk}

\newtheorem{teo}{Theorem}

\newtheorem{coro}{Corollary}

\newtheorem{exa}{Example}
\title{Fuzzy Group Identification Problems}

\author[1,3]{Federico Fioravanti \footnote{federico.fioravanti9@gmail.com (corresponding author)}}
\author[1,2]{Fernando Tohm\'e \footnote{ftohme@criba.edu.ar}}

\affil[1]{Instituto de Matem\'atica de Bah\'ia Blanca, CONICET - UNS, Bah\'{\i}a Blanca, Argentina}
\affil[2]{Departamento de Econom\'ia, Universidad Nacional del Sur, Bah\'ia Blanca, Argentina}
\affil[3]{Departamento de Matem\'atica, Universidad Nacional del Sur, Bah\'ia Blanca, Argentina}
\date{}
\begin{document}
\maketitle
\begin{abstract}
The Group Identification Problem (``Who is a J?'') introduced by Kasher and Rubinstein (1997) assumes a finite class of agents, each one with an opinion about the membership to a group $J$ of the members of the society, consisting in a function that indicates for each agent, including herself, the degree of membership  to  $J$. The problem is that of aggregating those functions, satisfying different sets of axioms and characterizing different {\em aggregators}.\\

The literature has already considered fuzzy versions of this problem. In this paper we consider alternative fuzzy presentations of the axiomatic of the original problem. While some results are analogous to those of the original crisp model, we show that our fuzzy version is able to overcome some of the main impossibility results of Kasher and Rubinstein.\\


\end{abstract}

\section{Introduction }
People usually classifies other people, objects or entities in groups.
Sometimes these classifications are obvious, as in the assignment of countries to the continent to which they belong.
But in many cases, opinions may not be even clear cut  and thus it becomes hard to reach a consensus. Usually human opinions involve vague concepts, and crisp assessments hardly formalize the inherent inaccuracies. 
For instance, if a group of people wants to identify which of them should be considered ``tall",  finding out  whether a member is such may be far from evident. Of course,  a consensus could be reached on that people that are 2 meters high or more are ``tall". But it is not clear whether someone whose height is 1.75 meters can be considered ``tall''.\\

This kind of classification problems, prone to vagueness and imprecision, gave the impetus for the introduction of fuzzy sets. Zadeh (1965) defined a fuzzy subset $U$ of a set $A$ as a membership function $f:A\rightarrow[0,1]$, where $f(a)$ indicates the degree of membership of $a$ in $U$. This allows to represent, in particular, preferences as degrees. In turn, there are several approaches to the aggregation of fuzzy preferences. Just to mention older contributions, (Dutta et al., 1987) deal with exact choices under vague preferences while (Dutta, 1986) investigates the structure of fuzzy aggregation rules determining fuzzy social orderings.\\

In this paper we re-examine a fuzzy approach to the group identification problem formalized by Kasher and Rubinstein (K-R) in ``Who is a J?"(1997). They consider a finite society that has to determine which one of its subsets of members consists of exactly those individuals that can be deemed to be $J$s. By a slight abuse of language, this subset is denoted $J$. Each different set of axioms postulated to yield a solution to this problem characterizes a class of aggregation functions, called Collective Identity Functions (CIFs). The Liberal one labels as a $J$ any individual that deems himself to be a $J$; the Dictatorial CIF is such that a single individual decides who is a $J$. Finally, the Oligarchic CIF determines that somebody is a $J$ if all the members of a given group agree on that.\\

After its introduction by K-R, the group identification problem in a crisp setting has been approached in many ways. Sung and Dimitrov (2005), refine the characterization of the liberal CIF, Miller (2008) studies the problem of defining more than two groups.  Saporiti and Cho (2017) analyze the incentives of voters in an incomplete preferences setting. Fioravanti and Tohm\'e (F-T) work with new axioms, finding new characterizations (2020a), and analyzing the problem in the case of infinite agents (2020b).\\

In this paper we assume that every agent, instead of labeling any of the individuals in the society as belonging or not to $J$, assigns a value in the $[0,1]$ interval to each agent, representing the degree in which that agent is believed to belong to $J$. We consider faithful fuzzy counterparts of axioms that can be found in the literature\footnote{Mainly in Kasher and Rubinstein (1997) and Fioravanti and Tohm\'e (2020a).}. Our goal in presenting faithful versions ot those crisp axioms is to remain as close as possible to the literature as to detect the notions in which the fuzzy characterizations lead to novel results. \\

We deal first with the axioms defining the Liberal aggregator, showing that the fuzzy versions of the results in K-R remain valid. But when we turn to the axioms defining the Dictatorial aggregator, our results differ from those in the crisp setting. More specifically, K-R prove an impossibility result when the domain and the range of the aggregator are restricted, indicating that the Dictatorial CIF is the only one satisfying those conditions. But there does not exist a clear ``translation'' of those restrictions into our framework, allowing different interpretations. In some of these we obtain more aggregators verifying the axioms, while in others there does not exist any of them.\\

The literature on fuzzy versions of aggregation problems has grown in the last years. Alcantud and D\'iaz (2017) study the notion of rationality in a fuzzy Arrovian framework, and analyze sequential fuzzy choice making. Duddy and Piggins (2018) prove the fuzzy counterpart of the Weymark's general oligarchy theorem (1984), while Raventos-Pujol et al. (2020) analyze Arrow's (1963) theorems in a fuzzy setting.\\

The group identification problem has also been considered from a fuzzy point of view. So, for instance, Cho and Park (2018) present a model of group identification for more than two groups, allowing fractional opinions, while Ballester and Garc\'{i}a-Lapresta (2008) deal with fuzzy opinions in a sequential model. An interesting paper, closely related to ours, is Alcantud and Andr\'{e}s Calle (2017),  who present a deep analysis of the aggregation problem of fuzzy opinions yielding fuzzy subsets. Their contribution introduces Fuzzy Collective Identity Functions (FCIFs), a concept that we adopt in this paper, defining the liberal and dictatorial aggregators, among others, presenting ways of circumventing  the classical impossibility result from K-R. The main difference between Alcantud and de Andr\'{e}s Calle's work and ours is that we find a characterization of the Liberal FCIF, and determine different domain and range conditions according to which the impossibility result can be either obtained or avoided. Later on, Alcantud and D\'{\i}az (2019) analyze the Liberal and Dictatorial aggregator in a fuzzy setting where a set of agents have to classify a set of objects in different groups (a more general framework than ours).\\

The plan of this paper is as follows. Section 2 is devoted to present the model and the set of axioms. In Section 3 we deal with Liberal aggregators while in Section 4 we focus on the Dictatorial aggregator. Finally, Section 5 concludes.

\section{Model and axioms}
Let $N=\{1,\ldots n\}$ be a set of agents that have to define who of them belongs to the group of $J$s. The opinion of agent $i$ is characterized by a function $p_{i}:N\rightarrow[0,1]$ where $p_{i}(j)$ indicates the assessment made by agent $i$ of the degree of membership of $j$ to $J$. Agent $i$ has thus a vector of opinions $P_{i}= (p_{i}(1),\ldots,p_{i}(n))$. A profile of opinions $P$ is a $n\times n$ matrix $P=(P_{1},\ldots,P_{n})$. With $\textbf{P}$ we denote the set of all the profiles of opinions. \\

A fuzzy subset $J$ of $N$ is characterized by a {\it membership (characteristic) function} $f_{J}:N\rightarrow [0,1]$, that indicates the degree of membership of an agent to $J$. Let $\textbf{F}=\{f_{J}:N\rightarrow [0,1]\}$ be the set of all possible membership functions of a fuzzy subset $J$. We denote by $\mathit{FJ}$ the {\it Fuzzy Collective Identity Function} (FCIF) such that $\mathit{FJ}:\textbf{P}\rightarrow \textbf{F}$.\\ 

A FCIF takes a profile of opinions and returns a membership function for the fuzzy subset $J$. More precisely, the membership function of the set $J$ associated to the profile $P\in\textbf{P}$ is denoted $f_{J}^{P}$. We do not impose further restrictions on this membership function. \\

In what follows we present in an axiomatic way the properties that a social planner would like to see implemented by a ``fair'' aggregation process. Most of them are fuzzy versions of properties already postulated by K-R in their seminal work. The following two axioms state that, if the opinion about an agent $i$ changes, say by increasing (decreasing) her degree of membership, then the aggregated degree of $i$ should at least (at most) the same as the previous one.  
\begin{itemize}
\item \textbf{Fuzzy Monotonicity} (FMON): let $P\in \textbf{P}$ be such that $f_{J}^{P}(i)=a$ and let $P'$ be a profile such that $P'_{k,i}> P_{k,i}$ for some $k$ with $P'_{h,j}= P_{h,j}$ for all $(h,j)\neq (k,i)$; then $f_{J}^{P'}(i)\geq a$.\\ 
\item \textbf{Fuzzy Strong Monotonicity} (FSMON): let $P\in \textbf{P}$ be such that $f_{J}^{P}(i)=a$ and let $P'$ be a profile such that $P'_{k,i}> P_{k,i}$ for some $k$ with $P'_{h,j}= P_{h,j}$ for all $(h,j)\neq (k,i)$,then $f_{J}^{P'}(i)>a$.\\ 
\end{itemize}
It is easy to see that if a FCIF verifies FSMON then it satisfies FMON.\\

The next axiom states that the aggregate opinion about an agent is bounded by the upper and lower bounds of all the individual opinions about that agent. That is:
\begin{itemize}
\item \textbf{Fuzzy Consensus} (FC): if for some \mbox{$a_{j}, b_{j}\in [0,1]$} and for all $i\in N$ $p_{i}(j)\geq a_{j}$ and $p_{i}(j)\leq b_{j}$, then $a_{j}\leq f_{J}^{P}(j)\leq b_{j}$.\footnote{Alcantud and de Andr\'{e}s Calle (2017) call this property {\em Unanimity}. Notice that if $a_{j}=0$ and $b_{j}=1$, every FCIF trivially satisfies this axiom.}
\end{itemize}
The following axiom states that if two agents are evaluated in a similar way, the FCIF must classify them also similarly:\footnote{SYM and FSYM will not be involved in the characterization results to be presented in Sections 3 and 4. Nevertheless, the difference between the crisp version presented in K-R and the fuzzy one is interesting enough to justify discussing them.}
\begin{itemize}
\item \textbf{Symmetry} (SYM): agents $j$ and $k$ are symmetric if:
\begin{itemize}
\item $p_{i}(j)=p_{i}(k)$ for all $i\in N-\{j,k\}$
\item $p_{j}(i)=p_{k}(i)$ for all $i\in N-\{j,k\}$
\item $p_{j}(k)=p_{k}(j)$
\item $p_{j}(j)=p_{k}(k)$
\end{itemize}

Then, if two agents $j$ and $k$ are symmetric it follows that $f_{J}^{P}(j)=f_{J}^{P}(k)$.
\end{itemize}
The following property is our first ``true'' fuzzy axiom. It considers an arbitrary threshold, $\alpha \in [0,1]$ that can be interpreted as indicating that, if an agent gets at least the degree $\alpha$, she is approved to belong to the group of $J$s.\footnote{While $\alpha=0.5$ represents majority approval voting, $\alpha$ can be given any other value.}
\begin{itemize}
\item \textbf{Fuzzy Symmetry} (FSYM): let $\alpha\in[0,1]$. Agents $j$ and $k$ are fuzzy-symmetric if:
\begin{itemize}
\item $p_{i}(j),p_{i}(k)\geq \alpha$ or $p_{i}(j),p_{i}(k)<\alpha$ for all $i\in N-\{j,k\}$
\item $p_{j}(i),p_{k}(i)\geq \alpha$ or $p_{j}(i),p_{k}(i)<\alpha$ for all $i\in N-\{j,k\}$
\item $p_{j}(k),p_{k}(j)\geq \alpha$ or $p_{j}(k),p_{k}(j)<\alpha$
\item $p_{j}(j),p_{k}(k)\geq \alpha$ or $p_{j}(j),p_{k}(k)<\alpha$ 
\end{itemize}
If two agents $j$ and $k$ are fuzzy-symmetric then $f_{J}^{P}(j),f_{J}^{P}(k)\geq \alpha$ or \mbox{$f_{J}^{P}(j),f_{J}^{P}(k)<\alpha$}.
\end{itemize}
While SYM means that $f_{J}^{P}(j)=f_{J}^{P}(k)$ for all $j,k\in N$, FSYM does not impose such a strong condition, since it prescribes that both agents must get either a higher or a lower degree. The following examples may be useful to understand the difference:
\begin{exa}
Consider $N=\{1,2\}$.
\begin{itemize}
\item The FCIF $f_J$ such that $f_{J}(i)=\frac{p_{1}(i)+p_{2}(i)}{2}$ satisfies SYM but not FSYM with $\alpha=0.5$. To see this, consider the profile $P=(\{0.1,0.9\},\{0.6,0.4\})$. We have that $f_{J}^{P}(1)=0.35<0.5$ and $f_{J}^{P}(2)=0.65>0.5$.
\item The FCIF $f_J$ such that $f_{J}(1)=0.2p_{1}(1)+0.2p_{2}(1)$ and $f_{J}(2)=0.1p_{1}(2)+0.1p_{2}(2)$ verifies FSYM if $\alpha=0.5$, but not SYM.
\end{itemize}
\end{exa}
The next axiom states that the aggregate opinion about an agent should only take into account the individual opinions about her. The downside is that its fuzzy version violates the idea that the aggregate opinion about an agent depends only on the opinions about her.\footnote{The good thing is that the results in the following section are as we expect them to be under both axioms.}.
\begin{itemize}
\item \textbf{Independence} (I): let $P$ and $P'$ be two profiles such that given an agent $j\in N$, $p_{i}(j)=p'_{i}(j)$ for all $i\in N$. Then $f_{J}^{P}(j)=f_{J}^{P'}(j)$.
\item \textbf{Fuzzy Independence} (FI):  let $\alpha\in[0,1]$. Let $P$ and $P'$ be two profiles such that given an agent $j\in N$ and for every $i\in N$ we have $p_{i}(j)\geq \alpha$ iff $p'_{i}(j)\geq \alpha$. Then $f_{J}^{P}(j)\geq \alpha$ iff $f_{J}^{P'}(j)\geq \alpha$.
\end{itemize}
The difference between the last two axioms is that while I means that $f_{J}(i)=f(p_{1}(i),\ldots,p_{n}(i))$, FI allows the aggregate opinion about any agent to be affected by the opinions held about other agents.\\

The next two examples show how these axioms affect the definition of the FCIF:
\begin{exa}
Consider $N=\{1,2\}$.
\begin{itemize}
\item The FCIF $f_J$ such that $f_{J}(i)=\frac{p_{1}(i)+p_{2}(i)}{2}$ verifies I but not FI with $\alpha=0.5$ as it can be seen in profiles $P=(\{0.1,0.6\},\{0.6,0.7\})$ and $P'=(\{0.45,0.6\},\{0.95,0.7\})$.\\
We have that $f_{J}^{P}(1)=0.35<0.5$ and $f_{J}^{P'}(1)=0.7>0.5$.
\item The FCIF $f_J$ such that $f_{J}(1)=0.25p_{1}(1)$$+0.25p_{1}(2)$ and $f_{J}(2)=0.25p_{2}(1)+$$0.25p_{2}(2)$ verifies FI with $\alpha=0.5$ but not I.
\end{itemize}
\end{exa}
The following axiom states that the opinion that an agent has about herself should be considered important to determine whether she is a $J$ or not.
\begin{itemize}
\item\textbf{Liberalism} (L): if $p_{i}(i)=1$ for some $i\in N$, then $f_{J}^{P}(k)=1$ for some $k\in N$.\\ 
If $p_{i}(i)=0$ for some $i\in N$, then $f_{J}^{P}(k)=0$ for some $k\in N$.
\end{itemize}
A fuzzy version of this axiom is:
\begin{itemize}
\item\textbf{Fuzzy Liberalism} (FL): let $\alpha\in[0,1]$. Then $p_{i}(i)\geq \alpha$ for some $i\in N$ if, and only if, $f_{J}^{P}(k)\geq \alpha$ for some $k\in N$. \\
\end{itemize}
The following example shows the difference between these two axioms:
\begin{exa}
Consider $N=\{1,\ldots,n\}$.
\begin{itemize}
\item The FCIF such that $f_{J}(i)=1$ if $p_{i}(i)=1$ and $f_{J}(i)=0$ otherwise, verifies L but not FL with $\alpha=0.5$.
\item The FCIF such that $f_{J}(i)=0.9$ if $p_{i}(i)\geq 0.5$ and $f_{J}(i)=0.1$ if  $p_{i}(i)< 0.5$, verifies FL with $\alpha=0.5$ but not L.
\end{itemize}
\end{exa} 
Other relevant properties concern the capacity that any agent may have of contributing to determine that any other is in $J$ (like parents with respect to the religious affiliation of their kids):

\begin{itemize}
\item \textbf{Extreme Liberalism} (EL):
\begin{itemize}

\item[(i)] If $p_{i}(j)=1$ for some $i,j\in N$, then $f_{J}^{P}(k)=1$ for some $k\in N$.
\item[(ii)] If $p_{i}(j)=0$ for some $i,j\in N$, then $f_{J}^{P}(k)=0$ for some $k\in N$.

\end{itemize}
\end{itemize}

A fuzzy version is:
\begin{itemize}
\item \textbf{Fuzzy Extreme Liberalism} (FEL): let $\alpha\in[0,1]$.
\begin{itemize}
\item[(i)] If $p_{i}(j)\geq \alpha$ for some $i,j\in N$, then $f_{J}^{P}(k)\geq \alpha$ for some $k\in N$. 
\item[(ii)] If $p_{i}(j)< \alpha$ for some $i,j\in N$, then $f_{J}^{P}(k)< \alpha$ for some $k\in N$.
\end{itemize}
\end{itemize}

Some final examples illustrate the difference between these last two axioms:
\begin{exa} 
Consider $N=\{1,\ldots,n\}$.
\begin{itemize}
\item A FCIF that satisfies EL but not FEL with $\alpha=0.5$: $f_J$ such that $f_{J}(i)=1$ for all $i\in N-\{1\}$ and $f_{J}(1)=0$ if there exist a $j,k\in N$ such that $p_{j}(k)=0$ or $f_{J}(1)=1$ otherwise.
\item A FCIF that verifies FEL with $\alpha=0.5$ but not EL: $f_{J}$ such that $f_{J}(1)=0.9$, $f_{J}(2)=0.1$ and $f_{j}(i)=0.5$ for all $i\in N-\{1,2\}$ and for every $P\in \textbf{P}$.
\end{itemize}
\end{exa} 
\section{Liberalism}
We define the Strong Liberal FCIF as:
$$ \mathit{L}(P_{1},\ldots,P_{N})(i)=L_{J}(i)=p_{i}(i)$$
\noindent for all $i\in N$.\\

It is straightforward to see that this FCIF is analogous to the Strong Liberal CIF that K-R introduce in their work. It verifies FMON, FC, FI and FL. Moreover, it is the only FCIF that satisfies this set of axioms:
\begin{teo}
The only FCIF that satisfies FMON, FC, FI and FL is the Strong Liberal FCIF.
\end{teo}
\begin{proof}
 Suppose, on the contrary, that there exists another FCIF satisfying FMON, FC, FI and FL, $f_j \neq L_J$. Consider then a profile $P$ to be such that $f_{J}^{P}(i)= b \neq p_i(i) = a$. $P$ exists since we have assumed that $f_j \neq L_J$. Suppose that $a >b$ (the case with $a < b$ can be analyzed analogously). Using FMON several times we can create a profile $P'$ identical to $P$ except that $p'_{j}(i)\leq b$ for all $j\neq i$, in such a way that $f_{J}^{P'}(i)\leq b$. Consider a profile $P''$ such that $p''_{i}(i)\geq a$ and $p''_{j}(k)\leq b$ for all $j,k\in N$ (except when $j=k=i$). By FC we have that $f_{J}^{P''}(k)\leq b$ for all $k\neq i$. So the set of agents $k \in N$ such that $f_{J}^{P''}(k)\geq a$ is either $\emptyset$ or consists just of $i$. Because of FL we have $f_{J}^{P''}(i)\geq  a$. But then we have a contradiction with FI, because agent $i$ is treated similarly (in the sense of FI) in profiles $P'$ and $P''$ but $f_{J}^{P'}(i)\leq b$ and $f_{J}^{P''}(i)\geq a$.
\end{proof}
This result is still valid even if we use a mixture of fuzzy and crisp axioms:
\begin{coro}
The only FCIF that verifies FMON, FC, I and L is the Strong Liberal FCIF.
\end{coro}
From the uniqueness of the Strong Liberal FCIF we obtain the following result:
\begin{coro}
If a FCIF verifies FMON, FC, FI and FL then it satisfies SYM and FSYM. 
\end{coro}
From the fact that the Liberal FCIF does not verify FSMON and that this axiom implies FMON we get:
\begin{coro}
There is no FCIF that satisfies FSMON, FC, FI and FL.
\end{coro}
The following two FCIFs are the fuzzy counterparts of the Unanimity and Inclusive CIF (from Fioravanti and Tohm\'e (2020a)):
\begin{itemize}
\item The {\em Unanimity} FCIF is defined as:
$$\mathit{U}(P_{1},\ldots,P_{N})(i)=U_{J}(i)=\min_{j}p_{j}(i)$$
\noindent for all $i\in N$.\footnote{Alcantud and de Andr\'{e}s Calle (2017) define this aggregator as the {\it Conjunctive} FCIF.}
\item The {\em Inclusive} FCIF is defined as:
$$\mathit{Inc}(P_{1},\ldots,P_{N})(i)=\mathit{Inc}_{J}(i)=\max_{j}p_{j}(i)$$
\noindent for all $i\in N$.\footnote{Alcantud and de Andr\'{e}s Calle (2017) call this aggregator the {\it Benevolent} FCIF.}
\end{itemize}
As in the crisp case, under extreme concepts of liberalism like EL or FEL, we obtain the same uniqueness results:
\begin{teo}
$\mathit{Inc}$ is the only FCIF that verifies FMON, FC, FI and EL (i) or FEL (i).\\ 
$\mathit{U}$ is the only FCIF that verifies FMON, FC, FI and EL(ii) or FEL(ii). 
\end{teo}
\begin{proof}
A similar construction as the used in the proof of Theorem 1 yields the proof of the two statements.
\end{proof}
We can derive the following impossibility result from Theorem 2:
\begin{coro}
There is no FCIF that verifies FMON, FC, FI and FEL or EL.
\end{coro}

\section{Dictatorship}
Kasher and Rubinstein use, in a section of their paper, a slightly modified version of the CIFs. They assume that there is a consensus in the society that there exists someone who is a $J$ and someone who is not a $J$.\\ 

Then, it follows that the only CIF (with the alternative domain and range conditions established for this case) that verifies Consensus and Independence is the Dictatorial one.\\

Here we define the Dictatorial FCIF with the agent $j$ as a dictator as:
$$ \mathit{D}(P_{1},\ldots,P_{N})(i)=D_{J}(i)=p_{j}(i)$$
\noindent for all $i\in N$.\\

There are many ways to interpret K-R's restrictions in our framework. One possibility is that in a profile $P_{i}$, there exists at least one $j$ such that $p_{i}(j)=1$ and one $k$ such that $p_{i}(k)=0$. We call this set of profiles $\textbf{P}^{\ast}$.\\

An alternative set of profiles is $\textbf{P}^{\ast\ast}$, in which for every agent $i$ there exists at least one $k$ and one $j$ such that $p_{i}(k)\geq \alpha$ and $p_{i}(j)\leq \alpha$.\\

Another possible case is that in which $p_{i}\neq \textbf{1}$ or $p_{i}\neq \textbf{0}$ for all $i\in N$. We call this set $\textbf{P}^{\ast\ast\ast}$.\\
With respect to the membership functions, we can consider the case in which there exist at least one $j$ and one $k$ such that $f_{J}(j)=1$ and $f_{J}(k)=0$. The class of such function is denoted $\textbf{F}^{\ast}$.\\

We define $\textbf{F}^{\ast\ast}$ as the set of membership functions such that for every profile $P$ there exists at least one $k$ and one $j$ such that $f_{J}^{P}(k)\geq \alpha$ and $f_{J}^{P}(k)\leq \alpha$.\\

Finally we have the family of functions such that $f_{J}\neq \textbf{1}$ or $f_{J}\neq \textbf{0}$. We call this set $\textbf{F}^{\ast\ast\ast}$.\\

It is easy to verify that:
$$\textbf{P}^{\ast}\subset \textbf{P}^{\ast\ast}\subset \textbf{P}^{\ast\ast\ast}$$ 
\noindent and
$$\textbf{F}^{\ast}\subset \textbf{F}^{\ast\ast}\subset \textbf{F}^{\ast\ast\ast}$$
A social planner may require different properties to be satisfied by the domain and range of membership functions. Depending on those specifications, there are various possibilities:\footnote{The domain and range restrictions introduce differences between FCIFs, even if they satisfy the same axioms. So, for instance, the Strong Liberal FCIF $L$ satisfies FC and FI, but it is not of the  \mbox{$FJ:\textbf{P}^{\ast\ast\ast}\rightarrow \textbf{F}^{\ast\ast\ast}$} type. Consider the following case, where $N=3$. Let $P=(\{1,0,0\},\{0,1,0\},\{0,0,1\})$, so we have that $P\in \textbf{P}^{\ast\ast\ast}$. But $L(i)=1$ for $i=1,2,3$, so we have that the range of $L$ is not necessarily included in $\textbf{F}^{\ast\ast\ast}$.}\footnote{The result for the FC and I case extends a conclusion advanced by Alcantud and de Andr\'es Calle (2017).}
\begin{teo}
	Consider FCIFs that satisfy axioms FC and FI or FC and I.	
	\begin{enumerate}
		\item  The Dictatorial is not the only FCIF such that $FJ: \textbf{P}^{\ast\ast\ast}\rightarrow\textbf{F}^{\ast\ast\ast}$.		
		\item The Dictatorial is the only FCIF such that $FJ: \textbf{P}^{\ast}\rightarrow\textbf{F}^{\ast\ast}$ or \mbox{$FJ: \textbf{P}^{\ast\ast}\rightarrow\textbf{F}^{\ast\ast}$}, and there is no FCIF such that $FJ:\textbf{P}^{\ast\ast\ast}\rightarrow\textbf{F}^{\ast\ast}$.		
		\item The Dictatorial is the only FCIF such that  $FJ: \textbf{P}^{\ast}\rightarrow\textbf{F}^{\ast}$, and there is no FCIF such that  $FJ: \textbf{P}^{\ast\ast}\rightarrow\textbf{F}^{\ast}$ or $FJ: \textbf{P}^{\ast\ast\ast}\rightarrow\textbf{F}^{\ast}$.
	\end{enumerate}
\end{teo}

\begin{proof} For simplicity, but without loss of generality, we will consider $\alpha=0.5$.
	\begin{enumerate}
		\item \begin{itemize}
			\item [(a)] We represent with $|P^{i}|^{>} $ the number of $p_{j}(i)$s that are larger than $0.5$ and with $|P^{i}|^{<} $ the number of those less than $0.5$ in a profile $P\in \textbf{P}$.\\		
			
		Now we consider the following FCIF:\\ 
		$$f_{J}^{P}(i)= \left\{ \begin{array}{cccccc}
		\frac{\min_{j}p_{j}(i)+\max_{j}p_{j}(i)}{2} &   if  & |P^{i}|^{>}=n  \\
		\\ \frac{0.5+\max_{j}p_{j}(i)}{2} &   if  & |P^{i}|^{>}>|P^{i}|^{<}\\
		\\ p_{j}(i) & if & p_{j}(i)=p_{k}(i) & for & all & j,k\in N \\
		\\0.5&   if  & |P^{i}|^{<}=|P^{i}|^{>}\\
		\\ \frac{0.5+\min_{j}p_{j}(i)}{2} &   if  & |P^{i}|^{<}>|P^{i}|^{>}\\
		\\\frac{\min_{j}p_{j}(i)+\max_{j}p_{j}(i)}{2} &   if  &|P^{i}|^{<}=n 
		\end{array}
		\right.$$\\		
		This FCIF verifies FC and FI and is not the Dictatorial FCIF.\\
		
		\item [(b)] We call a FCIF {\em Democratic} if
		$$f_{J}(i)=\frac{p_{1}(i)+\ldots+p_{n}(i)}{n}$$ 
		\noindent for all $i\in N$.\\
		
		This FCIF verifies FC and I.
		\end{itemize}		
		\item The proof is for the case where the FCIF verifies FC and FI (the other case is analogous).\\
		
	    We denote with $|f_{J}^{P}|^{>} $ the number of $f_{J}^{P}(i)$s larger or equal to $0.5$ while $|f_{J}^{P}|^{<} $ are those less than $0.5$ in the membership function $f\in \textbf{F}$.\\		
	    
		We say that a coalition $L\subseteq N$ is {\em fuzzy semidecisive} for agent $i$, if the following conditions are satisfied for every profile $P\in \textbf{P}$:\\		
		\begin{equation}
		{[\mbox{for all} \ j\in L,\; p_{j}(i)\geq 0.5\; and\; \mbox{for all} \ j\notin L,\; p_{j}(i)<0.5]\Rightarrow f_{j}^{P}(i)\geq 0.5}
		\end{equation} 
		\noindent and
		\begin{equation}
		[\mbox{for all} \ j\in L,\; p_{j}(i)< 0.5\; and\; \mbox{for all} \ j\notin L,\; p_{j}(i)>0.5]\Rightarrow f_{j}^{P}(i)< 0.5. 
		\end{equation}	
		A coalition $L\subseteq N$ is called {\em fuzzy semidecisive} if it is fuzzy semidecisive for every agent $i$ in $N$.\\	
			
		Analogously, we say that $L\subseteq N$ is {\em fuzzy decisive} over agent $i$ if the following conditions are satisfied for every profile $P\in \textbf{P}$:\\		
		$$ [\mbox{for all} \ j\in L,\; p_{j}(i)\geq 0.5]\Rightarrow f_{j}^{P}(i)\geq 0.5 $$		
		\noindent and		
		$$ [\mbox{for all} \ j\in L,\; p_{j}(i)< 0.5]\Rightarrow f_{j}^{P}(i)< 0.5. $$		
		In the same way, $L\subseteq N$ is said {\it fuzzy decisive} if it is fuzzy decisive for every agent $i$ in $N$.\\		
		
		We first prove the existence of a semidecisive coalition for an agent $i$, and then show that is semidecisive for all $i\in N$. Without loss of generality we assume that $N=3$ (the result extends easily to all other cases).\\		
		
		Consider, also w.l.g., the profile 		
		$$P^{1}=(\{0.1,0.6,0.1\},\{0.9,0.2,0.3\},\{0.1,0.2,0.8\}).$$ 		
		By hypothesis, $|f_{J}^{P^{1}}|^{>}\neq 3 $ and $|f_{J}^{P^{1}}|^{<}\neq 3$.\\		
		
		Suppose $|f_{J}^{P^{1}}|^{>}=2$, and that $f_{J}^{P^{1}}(1),f_{J}^{P^{1}}(2)\geq0.5$.
		Now consider the profile 	
		$$P^{2}=(\{0.1,0.6,0.6\},\{0.9,0.2,0.5\},\{0.1,0.2,0.8\}).$$		
		By FI, \mbox{$f_{J}^{P^{2}}(1),f_{J}^{P^{2}}(2)\geq0.5$} and by FC $f_{J}^{P^{2}}(3)\geq 0.5$, a contradiction.\\	
			
		Then $|f_{J}^{P^{1}}|^{>}=1$.\\		
		
		Suppose now that $f_{J}^{P^{2}}(2)\geq 0.5$. By FI, 		
		\begin{equation}
		\mbox{for all} \ P\in\textbf{P}, [p_{1}(2)\geq0.5,p_{2}(2)<0.5,p_{3}(2)<0.5]\Rightarrow f_{J}^{P}(2)\geq0.5.
		\end{equation}		
		It follows that (1) is verified by $L=\{1\}$ for $i=2$.\\	
			
		Consider now the profile  $$P^{3}=(\{0.6,0.1,0.2\},\{0.1,0.7,0.1\},\{0.1,0.8,0.4\})$$		
		\noindent and suppose by contradiction that $f_{J}^{P^{3}}(2)\geq 0.5$.\\ 	
			
		By FC, \mbox{$f_{J}^{P^{3}}(3)<0.5$}. Moreover, $f_{J}^{P^{3}}(1)<0.5$.\\	
			
		Otherwise, if $f_{J}^{P^{3}}(1)\geq0.5$, the profile $$P^{4}=(\{0.6,0.1,0.7\},\{0.1,0.7,0.9\},\{0.1,0.8,0.9\})$$
		 would lead to a contradiction, because by FI, $f_{J}^{P^{4}}(1),f_{J}^{P^{4}}(2)\geq 0.5$; and by FC, $f_{J}^{P^{4}}(3)> 0.5$.\\	
		 	
		Thus, the only remaining possibility is that $f_{J}^{P^{4}}(2)>0.5$ and by FI,		
		$$\mbox{for all} \ P\in\textbf{P}, [p_{1}(2)<0.5,p_{2}(2)>0.5,p_{3}(2)>0.5]\Rightarrow f_{J}^{P}(2)>0.5.$$		
		Consider the following profile 
				$$P^{5}=(\{0.2,0.1,0.7\},\{0.6,0.2,0.1\},\{0.6,0.1,0.3\}).$$		
		By FC $f_{J}^{P^{5}}(2)<0.5$.\\		
		
		If $f_{J}^{P^{5}}(3)\geq 0.5$, the profile 		
		$$P^{6}=(\{0.6,0.1,0.7\},\{0.6,0.8,0.1\},\{0.6,0.8,0.3\})$$		
		\noindent yields  $|f_{J}^{P^{6}}|^{>}=3$, a contradiction.\\	
			
		Then, the only possibility is that $f_{J}^{P^{5}}(1)\geq 0.5$ and by FI, 		
		\begin{equation}\mbox{for all} \ P\in\textbf{P}, [p_{1}(1)<0.5,p_{2}(1)\geq0.5,p_{3}(1)\geq 0.5]\Rightarrow f_{J}^{P}(1)\geq 0.5.
		\end{equation}		
		Finally, consider		
		$$P^{7}=(\{0.1,0.6,0.7\},\{0.6,0.3,0.9\},\{0.6,0.1,0.7\}).$$		
		By (3) and (4) we have that $f_{J}^{P^{7}}(1),f_{J}^{P^{7}}(2)\geq 0.5$ and by FC $f_{J}^{P^{7}}(3)\geq 0.5$, a contradiction.\\		
		
		Thus (2) is verified for agent $1$, who is fuzzy semidecisive for agent $2$.\\		
		
		Now, we can prove that if there exists a fuzzy decisive coalition $L\subseteq N$ for some $i\in N$, then $L$ is fuzzy semidecisive. Without loss of generality, we can suppose that $L = \{1\}$ is fuzzy decisive for $2$.\\	
			
		Let 
		$$P^{8}=(\{0.4,0.1,0.7\},\{0.6,0.2,0.1\},\{0.6,0.3,0.3\}).$$		
		By FC $f_{J}^{P^{8}}(2)<0.5$.\\	
			
		Moreover, if only $f_{J}^{P^{8}}(1)\geq 0.5$, then by FI, FC and the fact that agent $1$ is fuzzy decisive over $2$, we have that $|f_{J}^{P^{9}}|^{>}=3$ with		
		$$P^{9}=(\{0.4,0.6,0.7\},\{0.6,0.2,0.9\},\{0.6,0.3,0.8\}),$$		
		\noindent a contradiction.\\	
			
		Then $f_{J}^{P^{8}}(3)\geq 0.5$ and by FI, 		
		$$\mbox{for all} \ P\in\textbf{P}, [p_{1}(3)\geq0.5,p_{2}(3)<0.5,p_{3}(3)<0.5]\Rightarrow f_{J}^{P}(3)\geq 0.5.$$		
		Consider the profile 		
		$$P^{10}=(\{0.6,0.1,0.3\},\{0.2,0.2,0.7\},\{0.2,0.3,0.9\}).$$		
		If $f_{J}^{P^{10}}(3)\geq 0.5$, then by FI, FC and the fact that agent $1$ is fuzzy decisive over $2$, we have that $|f_{J}^{P^{11}}|^{>}=3$ with 		
		$$P^{11} =(\{0.7,0.6,0.1\},\{0.6,0.2,0.9\},\{0.6,0.3,0.8\}),$$		
		\noindent a contradiction.\\
				
		Thus $f_{J}^{P^{10}}(3)<0.5$ and by FI, 		
		$$\mbox{for all} \ P\in\textbf{P}, [p_{1}(3)<0.5,p_{2}(3)\geq0.5,p_{3}(3)\geq 0.5]\Rightarrow f_{J}^{P}(3)<0.5.$$		
		That is, agent $1$ is fuzzy semidecisive over $3$.\\	
			
		Using the same argument, agent $1$ is fuzzy semidecisive over herself. Then by definition $L= \{1\}$ is fuzzy semidecisive.\\		
		
		Now we show that the intersection of two fuzzy semidecisive coalitions is fuzzy semidecisive. First we prove that $L\cap L^{'} \neq\emptyset$. Suppose that this is not the case.\\		
		
		Without loss of generality, let $L=\{1\}$ y $L^{'} =\{2,3\}$.\\		
		
		Consider the profile 		
		$$P=(\{0.7,0.1,0.1\},\{0.4,0.7,0.1\},\{0.1,0.9,0.3\}).$$		
		Then $f_{J}^{P}(1),f_{J}^{P}(2)\geq 0.5$.\\	
			
		But by FI and FC we have $|f_{J}^{P^{'}}|^{>}=3$ with 		
		$$P^{'} =(\{0.7,0.1,0.7\},\{0.3,0.8,0.9\},\{0.1,0.7,0.8\}),$$		
		\noindent a contradiction.\\		
		
		Then $L\cap L^{'} \neq\emptyset$.\\		
		
		Second we prove that $L\cap L^{'}$ is fuzzy semidecisive. Without loss of generality, let $L=\{1,3\}$ and $L^{'}=\{1,2\}$.\\	
			
		Consider another profile 		
		$$P=(\{0.7,0.1,0.1\},\{0.4,0.1,0.8\},\{0.1,0.9,0.3\})$$		
		\noindent and suppose that, $f_{J}^{P}(1)<0.5$.\\		
		
		If $f_{J}^{P}(2)\geq 0.5$, then by FI, FC and the fact that $L^{'}$ is fuzzy semidecisive, we have that $|f_{J}^{P^{'}}|^{>}=3$ with 		
		$$P^{'} =(\{0.7,0.1,0.7\},\{0.8,0.1,0.9\},\{0.1,0.7,0.8\}),$$		
		\noindent a contradiction.\\	
			
		Alternatively, if $f_{J}^{P}(3)\geq 0.5$, then consider the following profile		
		$$P^{''} =(\{0.7,0.6,0.1\},\{0.8,0.1,0.8\},\{0.1,0.9,0.3\}).$$		
		By FI, $f_{J}^{P^{''}}(3)\geq 0.5$. \\	
			
		Because $L$ and $L^{'}$ are fuzzy semidecisive, $f_{J}^{P''}(1),f_{J}^{P^{''}}(2)\geq 0.5$.\\		
		
		Then $|f_{J}^{P^{''}}|^{>}=3$, a contradiction.\\	
			
		Thus $ f_{J}^{P}(1)\geq 0.5$ and by FI 		
		\begin{equation}\mbox{for all} \ P\in\textbf{P}, [p_{1}(1)\geq 0.5,p_{2}(1)<0.5,p_{3}(1)<0.5]\Rightarrow f_{J}^{P}(1)\geq 0.5.\end{equation}		
		Now consider another profile 		
		$$P=(\{0.4,0.1,0.7\},\{0.8,0.4,0.1\},\{0.8,0.2,0.3\}).$$		
		If $f_{J}^{P}(1)\geq 0.5$, then it follows from FI and the fact that $L$ and $L^{'}$ are fuzzy semidecisive, that $|f_{J}^{P^{'}}|^{>}=3$ if 		
		$$P^{'}=(\{0.4,0.6,0.7\},\{0.8,0.6,0.1\},\{0.8,0.2,0.8\}),$$		
		\noindent a contradiction.\\	
			
		Thus $f_{J}^{P}(1)<0.5$ and by FI 		
		\begin{equation} \mbox{for all} \ P\in\textbf{P}, [p_{1}(1)<0.5,p_{2}(1)\geq0.5,p_{3}(1)\geq0.5]\Rightarrow f_{J}^{P}(1)<0.5.\end{equation}		
		Then we have that $L = \{1\}$ is fuzzy semidecisive.\\	
			
		We can prove that given a coalition $L\subseteq N$, either $L$ is fuzzy semidecisive or $N\setminus L$ is fuzzy semidecisive.\\		
		
		By FC, $N$ is fuzzy semidecisive over $N$. Without loss of generality, fix $L=\{1,2\}$ and suppose by contradiction that $L$ is not fuzzy semidecisive. Then it must exist a profile $P$ and an individual $i\in N$, such that $p_{1}(i)\geq0.5,\ p_{2}(i)\geq 0.5,\ p_{3}(i)<0.5$ and $f_{J}^{P}(i)<0.5$; or $p_{1}(i)<0.5,\ p_{2}(i)<0.5,\ p_{3}(i)\geq0.5$ and $f_{J}^{P}(i)\geq 0.5$.\\	
			
		Suppose the latter is the case, and let $i=1$.  By FI, 		
		\begin{equation}
		\mbox{for all} \ P\in\textbf{P}, [p_{1}(1)<0.5,p_{2}(1)<0.5,p_{3}(1)\geq 0.5]\Rightarrow f_{J}^{P}(1)\geq 0.5.\end{equation}		
		We want to prove that $N\setminus L=\{3\}$ is fuzzy semidecisive for agent $1$.  To do this, we have to show that the following is the case: 		
		\begin{equation}
		\mbox{for all} \ P\in\textbf{P}, [p_{1}(1)\geq0.5,\;p_{2}(1)\geq 0.5,\;p_{3}(1)<0.5]\Rightarrow f_{J}^{P}(1)<0.5.\end{equation}		
		Consider the profile 		
		$$P^{1}=(\{0.6,0.1,0.1\},\{0.8,0.4,0.1\},\{0.4,0.6,0.3\}).$$ 		
		If $f_{J}^{P^{1}}(1)\geq 0.5$, the wanted result follows from FI. If not, i.e., if $f_{J}^{P^{1}}(1)<0.5$, we proceed as follows. 
		First, we note that by FC, $f_{J}^{P^{1}}(3)<0.5$. Second, if $f_{J}^{P^{1}}(2)\geq 0.5$, then we have that $|f_{J}^{P^{2}}|^{>}=3$ if 		
		$$P^{2}=(\{0.6,0.3,0.7\},\{0.8,0.1,0.9\},\{0.2,0.6,0.8\}).$$ 		
		Then it must be that case that $f_{J}^{P^{1}}(1)\geq 0.5$, and by FI, 		
		\begin{equation} \mbox{for all} \ P\in\textbf{P}, [p_{1}(1)\geq0.5,p_{2}(1)\geq 0.5,p_{3}(1)<0.5]\Rightarrow f_{J}^{P}(1)\geq 0.5.\end{equation}		
		Consider profile 		
		$$P^{3}=(\{0.4,0.1,0.7\},\{0.3,0.4,0.9\},\{0.4,0.6,0.3\}).$$ 		
		By FC $f_{J}^{P^{3}}(1)<0.5$.\\		
		If $f_{J}^{P^{3}}(2)\geq 0.5$, then by FC and FI we have that $|f_{J}^{P^{4}}|^{>}=3$ if 		
		$$P^{4}=(\{0.6,0.3,0.7\},\{0.8,0.1,0.9\},\{0.2,0.6,0.8\}).$$ 		
		Thus, it must be $f_{J}^{P^{3}}(3)\geq 0.5$; and by FI		
		\begin{equation}\mbox{for all} \ P\in\textbf{P}, [p_{1}(3)\geq 0.5,p_{2}(3)\geq 0.5,p_{3}(3)<0.5]\Rightarrow f_{J}^{P}(3)\geq 0.5.\end{equation}		
		But then, by FI, (7) y (10), we have that $|f_{J}^{P^{5}}|^{>}=3$ if 		
		$$P^{5}=(\{0.1,0.6,0.7\},\{0.3,0.6,0.9\},\{0.8,0.6,0.1\}),$$		
		\noindent a contradiction. \\	
			
		Then we have that $N\setminus L=\{3\}$ is fuzzy semidecisive for agent $1$. Finally we obtain that $L=\{3\}$ is fuzzy semidecisive for $N$.\\
				
		Now we prove that if a coalition $L$ is fuzzy semidecisive, with $|L|>1$, the subsets and supersets of $L$ are also fuzzy semidecisive.\\	
			
		Let $L\subseteq L^{'}\subseteq N$, with $L$ fuzzy semidecisive. If $L^{'}$ is not fuzzy semidecisive, then $N\setminus L^{'}$ is fuzzy semidecisive.\\
				
		But then $L\subseteq L^{'}$, $(N\setminus L^{'})\cap L=\emptyset$, contradicting our previous proof.\\	
			
		Thus $L^{'}$ must be fuzzy semidecisive over $N$.\\		
		
		Now consider $h\in L$. If $L\setminus \{h\}$ is fuzzy semidecisive, the result is proved. If not, we have that $N\setminus (L\setminus \{h\})$ is fuzzy semidecisive.\\		
		
		Then we have that $N\setminus (L\setminus \{h\})\cap L=\{h\}$.\\	
			
		The next step is to prove that there always exists an agent $h\in N$ such that $\{h\}$ is fuzzy semidecisive.\\	
			
		By FC, $N$ is fuzzy semidecisive, thus there exists $L^{'}\subseteq N$ such that $N\setminus L^{'}$ is fuzzy semidecisive.
		Then there must exist $L^{''}$ such that $(N\setminus L^{'})\setminus L^{''}$ is fuzzy semidecisive.\\	
			
		Because $N$ is finite, by iterating this argument we find an $h\in N$ that is fuzzy semidecisive over $N$.\\	
			
		Finally we prove that if a coalition $L\subseteq N$ is fuzzy semidecisive, then it is fuzzy decisive. For that, consider a fuzzy semidecisive coalition $L\subseteq N$.\\	
			
		Then there exists $h\in L$ that is fuzzy semidecisive over $N$. Without loss of generality, suppose that $h=1$.
		Suppose that $\{1\}$ is not fuzzy decisive for agent $2$. \\	
			
		Then it must exist a profile $P$ such that		
		\begin{itemize}
		\item[(a)]  $p_{1}(2)\geq 0.5$ and $f_{J}^{P}(2)<0.5$ or
		\item[(b)]  $p_{1}(2)<0.5$ and $f_{J}^{P}(2)\geq 0.5$.
		\end{itemize}		
		Suppose that (a) is the case (the other case is analogous). Since by hypothesis $\{1\}$ is fuzzy semidecisive for agent $2$, there has to exist a $j\neq 1$ such that $p_{j}(2)>0.5$.\\	
			
		Moreover, there must exist an agent $k\in N\setminus\{1,j\}$ such that $p_{k}(2)<0.5$.	Otherwise, by FC we get $f_{J}^{P}(2)\geq 0.5$. Without loss of generality, consider the case where 		
		$$P=(\{0.1,0.8,0.2\},\{0.1,0.1,0.9\},\{0.2,0.6,0.1\}).$$ 		
		By FC, $f_{J}^{P}(1)<0.5$. Then $f_{J}^{P}(3)\geq 0.5$.\\	
			
		By FI 		
		\begin{equation} \mbox{for all} \ P\in\textbf{P}, [p_{1}(3)<0.5,p_{2}(3)\geq 0.5,p_{3}(3)\geq 0.5]\Rightarrow f_{J}^{P}(3)\geq 0.5.\end{equation}		
		Consider the following profile 		
		$$P^{'}=(\{0.1,0.3,0.8\},\{0.9,0.1,0.2\},\{0.2,0.6,0.9\}).$$		
		If $f_{J}^{P^{'}}(3)<0.5$, then by FI 		
		\begin{equation} \mbox{for all} \ P\in\textbf{P}, [p_{1}(3)\geq 0.5,p_{2}(3)<0.5,p_{3}(3)\geq 0.5]\Rightarrow f_{J}^{P}(3)<0.5.\end{equation}		
		Because $\{2\}$ is fuzzy semidecisive over $3$, by (11) and (12), we get that it is fuzzy semidecisive over $N$.\\		
		But, $\{1\}\cap\{2\}=\emptyset$, a contradiction.\\	
			
		Then $f_{J}^{P'}(3)\geq 0.5$.\\		
		
		Moreover, by FI 		
		$$\mbox{for all} \ P\in\textbf{P}, [p_{1}(3)\geq 0.5,p_{2}(3)<0.5,p_{3}(3)\geq 0.5]\Rightarrow f_{J}^{P}(3)\geq 0.5.$$  
		Because $\{1\}$ is fuzzy semidecisive over $3$, we also know that 
		$$\mbox{for all} \ P\in\textbf{P}, [p_{1}(3)\geq 0.5,p_{2}(3)<0.5,p_{3}(3)<0.5]\Rightarrow f_{J}^{P}(3)\geq 0.5.$$		
		Then if for all $P\in \textbf{P}$, 		
		$$[p_{1}(3)\geq 0.5,\;p_{2}(3)<0.5,\;p_{3}(3)<0.5]\Rightarrow f_{J}^{P}(3)\geq 0.5,$$		
		\noindent we would get the desired result and $\{1\}$ would be fuzzy decisive for agent $3$. \\		
		
		Otherwise, we can repeat the previous argument and show that $\{3\}$ is fuzzy semidecisive over $N$, which would lead us to a contradiction since $\{1\}\cap\{3\}=\emptyset$.\\	
			
		Because $1\in L$, we have that $L$ is fuzzy semidecisive for agent $3$. \\	
			
		Finally, a similar argument shows that $L$ is fuzzy decisive for all $i\in N$.\\		
		
		So we have that whenever a FCIF verifies FC and FI, there must exist an agent fuzzy decisive over $N$.\\	
			
		The impossibility of a FCIF is proved by set inclusion, since the Dictatorial FCIF does not verify the hypothesis in $\textbf{P}^{\ast\ast\ast}\rightarrow\textbf{F}^{\ast\ast}$.		
		\item The proof is similar to (2).
	\end{enumerate}
\end{proof}
We see in the following example that the aggregator proposed in part 1a of Theorem 3 is not the Dictatorial FCIF.\\
\begin{exa}
If we consider the following profile $P=(\{0.2,0.3,0.9\},\{0.9,0.9,1\},$ $\{1,0,0.3\})$, we obtain $f_{J}^{P}(1)=0.75$, $f_{J}^{P}(2)=0.35$ and $f_{J}^{P}(3)=0.75$, which is different from the opinion of any agent.
\end{exa}
The election of the fuzzy structure and even the election between the fuzzy or crisp version of the Independence axiom provides the opportunity to find non dictatorial rules, something that is not possible in the case that K-R analyze.\\ 

\textit{Remark about the fuzzy consensus axiom:} A function $f:\mathbb{R}^{k}\rightarrow \mathbb{R}$ that verifies $\min\{a_{1},\ldots,a_{k}\}\leq f(a_{1},\ldots,a_{k})\leq \max\{a_{1},\ldots,a_{k}\}$ is called a $k$-{\em dimensional mean} (Hajja, 2013). A FCIF that verifies FC is then, a $n$-dimensional mean.\\ 

The mean can yield different levels of ``democracy'', in the sense that the Dictatorial FCIF is a mean while the Democratic FCIF is also a mean.\\

A natural question is whether every $n$-dimensional mean yields a fuzzy aggregation function.\\

The following example shows that the answer is negative.
\begin{exa}
{\em The Inclusive and Unanimous FCIFs are $n$-dimensional means.\\ 
Consider the Unanimous FCIF and the profile 
$$P=(\{0,0.2,0.3\},\{1,0,0.5\},\{0.3,0.8,0\})$$ 
We obtain $U(i)=0$ for all $i\in N$, a membership function that does not belong to $F^{\ast\ast\ast}$.}
\end{exa}

\section{Conclusion}
In this work we present a new analysis, from a fuzzy point of view, of the Group Identification Problem. The opinions of the agents are no longer crisp statements about the membership or not to a group. Instead of that, their opinions are expressed in terms of degrees of membership to the class of $J$s.\\

We present the axioms that have already been analyzed in the literature and work with fuzzy versions of them. In the case of `Liberal' aggregators, the uniqueness and impossibility results from K-R and F-T still remain.\\

The axioms FL and L are very restrictive, and do not allow more rules other than the Strong Liberal FCIF. They even yield an impossibility result when we consider such a strong axiom as FSMON. With the FEL and EL axioms, we obtain the same results as in the crisp version.\\

When we deal with the Dictatorial aggregator, the results are richer than in the crisp version, due to the different interpretations that we may give to the domain and range conditions postulated in K-R.\\

Depending on the goals of the social planner we can have several, just one or no rules satisfying the desired properties.\\

The binary nature of the determination of whether someone has more or less than an $\alpha$ degree of acceptance, leads in general to the preservation of the uniqueness and impossibility results. However, when we modify the domain and range of the FCIFs, we get the possibility of designing more rules than just the Dictatorial one.

\end{document}